\documentclass[11pt]{article}

\usepackage{fullpage}

\usepackage{times}
\usepackage{helvet}
\usepackage{courier}
\usepackage{amssymb,amsmath}
\usepackage{comment}
\usepackage{graphicx}

\usepackage{stmaryrd}

\usepackage{amsmath}
\usepackage{comment}
\usepackage{amssymb}

\usepackage{xcolor}
\usepackage{comment}

\usepackage[T1]{fontenc}
\usepackage{charter}
\usepackage[expert]{mathdesign}

\newcommand{\ignore}[1]{}

\begin{document}

\title{Algebraic Global Gadgetry for Surjective Constraint Satisfaction}

\author{Hubie Chen\\King's College London\\\texttt{hubie.chen@kcl.ac.uk}}

\date{ } 

\maketitle

\newcommand{\confversion}[1]{}
\newcommand{\longversion}[1]{#1}

\newtheorem{theorem}{Theorem}[section]
\newtheorem{conjecture}[theorem]{Conjecture}
\newtheorem{corollary}[theorem]{Corollary}
\newtheorem{proposition}[theorem]{Proposition}
\newtheorem{prop}[theorem]{Proposition}
\newtheorem{lemma}[theorem]{Lemma}
\newtheorem{remarkcore}[theorem]{Remark}
\newtheorem{exercisecore}[theorem]{Exercise}
\newtheorem{examplecore}[theorem]{Example}
\newtheorem{assumptioncore}[theorem]{Assumption}

\newenvironment{assumption}
  {\begin{assumptioncore}\rm}
  {\hfill $\Box$\end{assumptioncore}}

\newenvironment{example}
  {\begin{examplecore}\rm}
  {\hfill $\Box$\end{examplecore}}

\newenvironment{exercise}
 {\begin{exercisecore}\rm}
 {\hfill $\Box$\end{exercisecore}}

\newenvironment{remark}
 {\begin{remarkcore}\rm}
 {\hfill $\Box$\end{remarkcore}}

\newenvironment{proof}{\noindent\textbf{Proof\/}.}{$\Box$ \vspace{1mm}}

\newtheorem{researchq}{Research Question}

\newtheorem{newremarkcore}[theorem]{Remark}

\newenvironment{newremark}
  {\begin{newremarkcore}\rm}
  {\end{newremarkcore}}

\newtheorem{definitioncore}[theorem]{Definition}

\newenvironment{definition}
  {\begin{definitioncore}\rm}
  {\end{definitioncore}}

\newcommand{\ppequiv}{\mathsf{PPEQ}}
\newcommand{\eq}{\mathsf{EQ}}
\newcommand{\iso}{\mathsf{ISO}}
\newcommand{\ppeq}{\ppequiv}
\newcommand{\ppiso}{\mathsf{PPISO}}
\newcommand{\boolppiso}{\mathsf{BOOL}\mbox{-}\mathsf{PPISO}}
\newcommand{\csp}{\mathsf{CSP}}
\newcommand{\scsp}{\mathsf{SCSP}}
\newcommand{\cond}{\mathsf{COND}}
\newcommand{\gi}{\mathsf{GI}}
\newcommand{\ci}{\mathsf{CI}}
\newcommand{\rela}{\mathbf{A}}
\newcommand{\relb}{\mathbf{B}}
\newcommand{\relc}{\mathbf{C}}
\newcommand{\reld}{\mathbf{D}}
\newcommand{\relg}{\mathbf{G}}
\newcommand{\relh}{\mathbf{H}}
\newcommand{\reln}{\mathbf{N}}

\newcommand{\rele}{\mathbf{E}}

\newcommand{\alga}{\mathbb{A}}
\newcommand{\algb}{\mathbb{B}}
\newcommand{\algab}{\mathbb{A}_{\relb}}

\newcommand{\idemp}{I}

\newcommand{\varv}{\mathcal{V}}
\newcommand{\variety}{\mathcal{V}}
\newcommand{\false}{\mathsf{false}}
\newcommand{\true}{\mathsf{true}}
\newcommand{\pol}{\mathsf{Pol}}
\newcommand{\inv}{\mathsf{Inv}}
\newcommand{\cc}{\mathcal{C}}
\newcommand{\alg}{\mathsf{Alg}}
\newcommand{\pitwo}{\Pi_2^p}
\newcommand{\sigmatwo}{\Sigma_2^p}
\newcommand{\pithree}{\Pi_3^p}
\newcommand{\sigmathree}{\Sigma_3^p}

\newcommand{\fancyg}{\mathcal{G}}
\newcommand{\tw}{\mathsf{tw}}

\newcommand{\mc}{\mathsf{MC}}
\newcommand{\mcs}{\mathsf{MC}_s}

\newcommand{\mcb}{\mathsf{MC_b}}

\newcommand{\qc}{\mathrm{QC}}
\newcommand{\normqc}{\mathrm{norm\mbox{-}QC}}
\newcommand{\rqc}{\mathsf{RQC\mbox{-}MC}}

\newcommand{\qcfo}{\mathrm{QCFO}}
\newcommand{\qcfofk}{\qcfo_{\forall}^k}
\newcommand{\qcfoek}{\qcfo_{\exists}^k}

\newcommand{\fo}{\mathrm{FO}}
\newcommand{\fok}{\mathrm{FO}^k}

\newcommand{\tup}[1]{\overline{#1}}

\newcommand{\nn}{\mathsf{nn}}
\newcommand{\bush}{\mathsf{bush}}
\newcommand{\width}{\mathsf{width}}

\newcommand{\un}{N^{\forall}}
\newcommand{\en}{N^{\exists}}

\newcommand{\ord}{\tup{u}}
\newcommand{\ordp}[1]{\tup{#1}}

\newcommand{\gc}{G^{-C}}

\newcommand{\ar}{\mathrm{ar}}
\newcommand{\free}{\mathsf{free}}
\newcommand{\vars}{\mathsf{vars}}

\newcommand{\qed}{}

\newcommand{\f}{\mathcal{F}}

\newcommand{\pow}{\wp}

\newcommand{\N}{\mathbb{N}}

\newcommand{\param}[1]{\mathsf{param}\textup{-}#1}

\newcommand{\dom}{\mathsf{dom}}

\newcommand{\org}{\mathrm{org}^+}
\newcommand{\lay}{\mathrm{lay}^+}

\newcommand{\und}[1]{\underline{#1}}

\newcommand{\clo}{\mathsf{closure}}

\newcommand{\thick}{\mathsf{thick}}
\newcommand{\thickl}{\thick_l}
\newcommand{\localthickl}{\mathsf{local}\textup{-}\thickl}
\newcommand{\quantthickl}{\mathsf{quant}\textup{-}\thickl}

\newcommand{\lowerdeg}{\mathsf{lower}\textup{-}\mathsf{deg}}

\newcommand{\restrict}{\upharpoonright}

\renewcommand{\nu}[1]{\textup{{\small $\mathsf{nu}$}-{$ #1 $}}}

\newcommand{\case}[1]{\textup{{\small $\mathsf{case}$}-{$ #1 $}}}

\newcommand{\coclique}{\textup{{\small  $\mathsf{co}$}-{\small $\mathsf{CLIQUE}$}}}
\newcommand{\clique}{\textup{{\small $\mathsf{CLIQUE}$}}}

\newcommand{\caseclique}{\case{\clique}}
\newcommand{\casecoclique}{\case{\coclique}}

\newcommand{\fpt}{\textup{\small $\mathsf{FPT}$}}
\newcommand{\wone}{\textup{\small $\mathsf{W[1]}$}}
\newcommand{\cowone}{\textup{\small $\mathsf{co}$-$\mathsf{W[1]}$}}

\renewcommand{\S}{\mathcal{S}}
\newcommand{\G}{\mathcal{G}}

\newcommand{\image}{\mathsf{image}}

\begin{abstract}
\begin{quote}
The constraint satisfaction problem (CSP)
on a finite
relational structure $\relb$ is to decide,
given a set of constraints on variables where the relations 
come from $\relb$, whether or not there is a assignment
to the variables satisfying all of the constraints;
the surjective CSP is the variant where one decides
the existence of a surjective satisfying assignment
onto the universe of $\relb$.

We present an algebraic framework for proving hardness
results on surjective CSPs; essentially, this framework
computes global gadgetry that permits one to present
a reduction from a classical CSP to a surjective CSP.
We show how to derive a number of hardness results
for surjective CSP in this framework, 
including the hardness of 
the disconnected cut problem,
of the no-rainbow 3-coloring problem,
and of the surjective CSP on all
2-element structures known to be intractable
(in this setting).
Our framework thus allows us to
unify these hardness results, 
and reveal common structure
among them; we believe 
that our hardness proof for the 
disconnected cut problem
is more succinct than the original.
In our view, the framework also makes very transparent
a way in which classical CSPs can be reduced to surjective CSPs.
\end{quote}
\end{abstract}

\section{Introduction}

\subsection{Background}

The \emph{constraint satisfaction problem (CSP)} is 
a computational problem in which one is to decide,
given a set of constraints on variables, whether or not
there is an assignment to the variables satisfying all of the
constraints.
This problem appears in many guises throughout computer science,
for instance, in database theory, artificial intelligence, and
the study of graph homomorphisms.
One obtains a rich and natural family of problems by
defining, for each relational structure $\relb$,
the problem $\csp(\relb)$ to be the case of the CSP where
the relations used to specify constraints must come from $\relb$.
Recall that a relational structure consists of a \emph{universe},
a set---assumed to be finite throughout this article---and 
indexed relations over this universe; 
we will refer to a relational structure whose universe has size
$k$ as a \emph{$k$-element structure}.
Let us refer to a problem of the form $\csp(\relb)$ as a
\emph{classical CSP}.  In a celebrated development,
Bulatov~\cite{Bulatov17-dichotomy} and Zhuk~\cite{Zhuk17-dichotomy}, independently, classified
the complexity of all classical CSPs, showing each to be either polynomial-time tractable or NP-complete; the algorithmic and
complexity aspects of this problem family, and their many links to 
numerous regions of theoretical computer science, had drawn significant
attention over the preceding two decades or so~\cite{KrokhinZivny17-dagstuhlfollowup}.

A natural variant of the CSP is the problem
where an instance is again a set of constraints,
but the question is to decide whether or not there is a
\emph{surjective} satisfying assignment to the variables.
For each relational structure $\relb$, 
in analogy to the definition of the problem $\csp(\relb)$,
one can define
the \emph{surjective CSP} over $\relb$, 
denoted by $\scsp(\relb)$, 
as the case of this variant where the relations used
to specify constraints must come from $\relb$; then,
the question is to decide if there exists
a satisfying assignment that is surjective onto
the universe of $\relb$.
Let us refer to a problem of the form $\scsp(\relb)$
as a \emph{surjective CSP}; each such problem is
readily seen to be in NP.
We remark that such a problem $\scsp(\relb)$ can 
be formulated as the problem of deciding, given as input a
relational structure $\rela$, whether or not
$\rela$ admits a \emph{homomorphism} to $\relb$
that is surjective onto the universe of $\relb$.
Surjective CSPs are studied and relevant in a number of contexts:
\begin{itemize}
\item In graph theory, there is a literature studying 
  the problems $\scsp(\relh)$, over graphs $\relh$ (see for example~\cite{BodirskyKaraMartin12-surjectivesurvey,GolovachPS12-surjective-trees,MartinPaulusma15-discut,GolovachJMPS19-surjective-new-hardness,LaroseMartinPaulusma19-surjective-reflexive-digraphs} and the references therein).
  Typically, in this literature such a problem is
  formulated as a homomorphism problem, and referred to
  as the \emph{surjective $H$-coloring} problem.

Here, we want to introduce a particular such problem,
the \emph{disconnected cut} problem: 
given a graph $G$, decide if it has a cut that is disconnected;
that is, decide if there exists a non-empty, proper subset $U$ of its
vertex set $V$ such that the induced subgraphs
$G[U]$ and $G[V\setminus U]$ are disconnected.
This problem arises in numerous contexts~\cite[Section 1]{MartinPaulusma15-discut},
including the theory of vertex cuts and graph connectivity,
graph partitions, and vertex covers of graphs;
it is equivalent to the problem $\scsp(\relh)$
where $\relh$ is the reflexive $4$-cycle.
This problem was proved to be NP-complete by Martin and Paulusma~\cite{MartinPaulusma15-discut}.

\item In the theory of hypergraphs, 
  colorings of hypergraphs with varied edge types
  have been studied.
  For example, 
  \emph{mixed hypergraphs}~\cite{Voloshin95-on-upper-chromatic-number,KralKPV06-coloring-mixed-hypertrees}
  consist of
  \emph{edges} and \emph{co-edges}; an assignment of colors 
  to vertices is considered to properly color an edge when
  two vertices therein are mapped to different colors,
  and to properly color a co-edge when 
  two vertices therein are mapped to the same color.
  
  Here, we wish to introduce the \emph{no-rainbow $3$-coloring} problem:
  given a hypergraph consisting of size $3$ co-edges,
  decide if there is a coloring with $3$ colors in its image.
  Equivalently, one is given a $3$-uniform hypergraph,
  and the question is whether or not there is a vertex coloring,
  with image size $3$, such that for each edge, 
  it is not the case that the three vertices therein are
  mapped to three different colors---that is, 
  no edge yields a rainbow.
  This problem is equivalent to the problem $\scsp(\reln)$
  where $\reln$ is the structure with universe $N = \{ 0, 1, 2 \}$
  and a single ternary relation 
  $\{ (a,b,c) \in N^3 ~|~ \{ a, b, c \} \neq N \}$;
  it was highlighted as an open problem in a survey on surjective CSPs~\cite{BodirskyKaraMartin12-surjectivesurvey},
  and this relation arises naturally in universal algebra
  within the theory of maximal clones~\cite{Pinsker02-maximal-clones}.
  This problem was recently proved to be NP-complete by 
  Zhuk~\cite{Zhuk20-norainbow}.

\item Surjective homomorphisms appear naturally in the linear-algebraic theory of homomorphism-related combinatorial quantities that was pioneered by Lov\'{a}sz~\cite{Lovasz67-operations-structures,Lovasz12-book-networks-and-limits,CurticapeanDellMarx17-homomorphisms-good-basis,ChenCurticapeanDell19-counting}.

\item In the context of constraint satisfaction, surjectivity is a form of \emph{global cardinality constraint};
such constraints are studied heavily in constraint processing
and in constraint programming, and have been investigated
from the complexity-theoretic viewpoint~\cite{SamerS11-global-cardinality-constraint,BulatovMarx10-global}.
\end{itemize}

The natural research program in this area
is to work towards the obtention of a full classification of all surjective CSPs,
akin to the classification of all classical CSPs.
In addition to providing information on the complexity 
of individual surjective CSPs of interest,
such a classification could provide a rich source of problems
from which one could try to reduce in developing NP-hardness
proofs.\footnote{
This has been done in the preceding literature: complexity hardness results on 
surjective CSPs were used crucially in a complexity classification of rooted phylogeny problems~\cite[Section 5]{BodirskyMueller11-phylogeny}.
Complexity hardness results on classical CSPs have long been used to derive further hardness results; indeed, we show hardness results on surjective CSPs in this article by reducing from classical CSPs!

}
  Moreover, the classification of classical CSPs
is appealing in that the criterion delineating the
two complexity behaviors~\cite{Bulatov17-dichotomy,Zhuk17-dichotomy,ChenLarose17-metaquestions}
is a relatively simple algebraic one which robustly
admits multiple natural characterizations 
(and which possesses relatively low complexity!);
one could hope for such a criterion in the context of
surjective CSPs.
Furthermore, studying the complexity of classical CSPs led to a rich general theory having interactions with many other areas of theoretical computer science, and one might hope for a similar situation
with surjective CSPs.  Regarding known classifications,
a dichotomy theorem that classifies the complexity 
of each surjective CSP on a $2$-element structure
has been known for many years~\cite[Theorem 6.12]{CreignouKhannaSudan01-boolean}, but a classification
for all $3$-element structures is currently open.
Here, we can mention that optimization and counting flavors
of the surjective CSP have also been studied~\cite{FullaUppmanZivny19-boolean-surjective-valued-csps,FockeGZ19-counting-surjhoms-and-compactions,ChenCurticapeanDell19-counting}.

A striking, frustrating, and somewhat mysterious aspect 
of the literature on surjective CSPs thus far is
the apparent difficulty of establishing intractability results;
this has been mentioned on a number of occasions (Sections 1 of~\cite{BodirskyKaraMartin12-surjectivesurvey,GolovachJMPS19-surjective-new-hardness,LaroseMartinPaulusma19-surjective-reflexive-digraphs}).
In particular, the various techniques used for proving intractability
seem disparate, ad hoc, and not obviously amenable to 
generalization.  As barometers of this claim, 
consider that the intractability of the \emph{disconnected cut}
problem was the subject of an entire article~\cite{MartinPaulusma15-discut},
and that this article as well as the one~\cite{Zhuk20-norainbow} exhibiting
the intractability of \emph{no-rainbow $3$-coloring} 
do not explicitly develop machinery nor technology that
permits the derivation of other hardness results in the literature.
The situation here can be contrasted with even the early work on
classical CSPs, where results tended to address
classes of relational structures~\cite{KrokhinBJ03-csp-short-survey}, 
as opposed to individual relational structures.
The present author did present an algebraic 
sufficient condition of hardness for surjective CSPs~\cite{Chen14-hardness-surjective-csp};
while this condition applies to each $2$-element structure
having a hard surjective CSP, it provably
does not apply to prominent examples such as \emph{disconnected cut} and \emph{no-rainbow $3$-coloring}.\footnote{
This inapplicability 
is proved in Section~\ref{subsect:diagonal-cautious}.
}

\subsection{Contributions}

What seems to be lacking in this research area is unified, general technology
for proving hardness results.
In this article, we make a contribution in this direction
by presenting a framework wherein one can establish 
reductions from classical CSPs to surjective CSPs.
We derive the following frontier results using our framework:
\begin{itemize}
\item The NP-hardness of the \emph{disconnected cut} problem (Section~\ref{subsect:disconn-cut}).
\item The NP-hardness of the \emph{no-rainbow $3$-coloring} problem (Section~\ref{subsect:no-rainbow}).
\item The NP-hardness of each problem $\scsp(\relb)$
that is intractable according to the mentioned dichotomy theorem
on $2$-element structures $\relb$ (Section~\ref{subsect:diagonal-cautious}).

\end{itemize}
We believe that our hardness proof for the \emph{disconnected cut} problem is
more succinct than the original, 
not just literally, but also in a conceptual sense;
 consider that the original proof centrally
involved reasoning about the \emph{diameter}
of produced instance graphs, whereas such reasoning is fully absent 
from our proof.  The third hardness result named here is 
obtained by showing that the 
present author's previous algebraic sufficient condition for hardness~\cite{Chen14-hardness-surjective-csp}
can be obtained within the framework of the present article.

Our framework (presented in Section~\ref{sect:framework}) can
be described on a high-level as follows.
Let $\relb$ be a structure.  
When the framework is applicable to 
the structure $\relb$, it permits,
for any finite set $V$ of variables,
the computation of an instance $\Phi_V$
of $\scsp(\relb)$ with the following properties:
each \emph{surjective} satisfying assignment of $\Phi_V$
encodes an assignment $g: V \to D$, 
where $D$ is a finite set (which is uniform over all $V$);
and, conversely, each assignment $g: V \to D$
is encoded by a surjective satisfying assignment of $\Phi_V$.
Thus, the instance $\Phi_V$, in essence, encodes
the space $G_{V,D}$ of all assignments $g: V \to D$.
Then, adding constraints to $\Phi_V$ places constraints
on the allowable assignments in this space $G_{V,D}$;
in this way, the problem $\scsp(\relb)$
can simulate---formally, admit a reduction from---a
problem $\csp(\reld)$, where $\reld$ is a
structure with universe $D$ and appropriately
defined relations.

We conceive of the instances $\Phi_V$ 
as a form of \emph{global gadgetry}
or of scaffolding that, once configured, 
enables the use of \emph{local gadgetry} 
encoding individual constraints.  
We believe that our framework makes very transparent
\emph{why} and \emph{how} surjective CSPs permit
reductions from classical CSPs;
also, we believe that
the \emph{roles} of the various introduced constraints
in these reductions is highly transparent.
We hypothesize that progress in this area was hindered by
the lack of systematic tools
for computing and reasoning about such global gadgetry.
We provide evidence for the well-behavedness of our framework, by showing that under a certain assumption, whether or not the framework is
applicable is decidable
(Section~\ref{sect:decidability}).

Our framework has implications beyond those described for the complexity of the decision problems
$\scsp(\relb)$.
We show (Section~\ref{sect:condensations})
that the framework can also be used to derive NP-hardness
of cases of the problem of deciding if
a given relational structure $\rela$ has a \emph{condensation}
to a second relational structure $\relb$,
where a condensation is a surjective satisfying assignment that also maps the tuples of each relation of $\rela$ surjectively onto
the corresponding relation of $\relb$.
We also discuss 
implications for the \emph{sparsifiability}
of the problems $\scsp(\relb)$ (Section~\ref{sect:sparsifiability}); this setting concerns polynomial-time reductions from a hard problem $\scsp(\relb)$ to a second problem, such that the output instance obeys a size guarantee.

Overall, we wish to emphasize our introduction of systematic machinery,
the unifying nature of this machinery in terms of yielding
a common understanding of previously disparate hardness proofs,
and the clarifying and transparent nature of our approach.
Our work also ties into a large body of work on 
CSPs and other problems that utilizes
algebraic closure properties of solution spaces 
as a way of obtaining insight into the complexity of problems.
In particular, we make use of such properties in the
form of so-called \emph{partial polymorphisms},
and
our work contributes to a growing literature which
uses partial polymorphisms as a tool to understand
the complexity of various computational problems~\cite{JonssonLNZ17-partial-clones-sat-problems,LagerkvistR17-inverse-sat,ChenJP18-sparsifiability,CouceiroHL19-fine-grained-partial-polymorphisms,LagerkvistW20-sparsification}.

\section{Preliminaries}

When $f: A \to B$ and $g: B \to C$ are mappings, we use
 $g \circ f$  to denote their composition.
 We adhere to the convention that for any set $S$,
there is a single element in $S^0$; 
this element is referred to as the \emph{empty tuple},
and is
denoted by $\epsilon$.

\subsection{Structures, formulas, and problems}

A \emph{signature} is a set of \emph{relation symbols};
each relation symbol $R$ has an associated arity (a natural number),
denoted by $\ar(R)$.
A \emph{structure} $\relb$ over signature $\sigma$
consists of a \emph{universe} $B$ which is a set,
and an interpretation $R^{\relb} \subseteq B^{\ar(R)}$
for each relation symbol $R \in \sigma$.
We tend to use the letters $\rela, \relb, \ldots$
to denote structures, and the letters $A, B, \ldots$
to denote their respective universes.
In this article, we 
assume that signatures under discussion are finite, and 
also assume that all structures under discussion are finite;
a structure is finite when its universe is finite.

By an \emph{atom} (over a signature $\sigma$), 
we refer to a formula of the form
$R(v_1, \ldots, v_k)$ where $R$ is a relation symbol (in $\sigma$), 
$k = \ar(R)$,
and the $v_i$ are variables.
An \emph{$\wedge$-formula} (over a signature $\sigma$) is a conjunction 
$\beta_1 \wedge \cdots \wedge \beta_m$
where each conjunct $\beta_i$ is an atom (over $\sigma$)
or a variable equality $v = v'$.  Here, we permit an empty conjunction.
With respect to a structure $\relb$,
an $\wedge$-formula $\phi$ over the signature of $\relb$
is \emph{satisfied} by an mapping $f$ to $B$
defined on the variables of $\phi$ when:
\begin{itemize}

\item for each atom $R(v_1, \ldots, v_k)$ in $\phi$,
it holds that $(f(v_1),\ldots,f(v_k)) \in R^{\relb}$,
and
\item for each variable equality $v = v'$, it holds that $f(v) = f(v')$.
  \end{itemize}
When this holds, we refer to $f$ as a \emph{satisfying assignment}
of $\phi$ (over $\relb$).  Note that the empty conjunction
is considered to be satisfied by any such mapping $f$.

We now define the computational problems to be studied.
For each structure $\relb$, define
$\csp(\relb)$ to be the problem of deciding, given 
a $\wedge$-formula $\phi$ (over the signature of $\relb$),
whether or not there exists a 
satisfying assignment, that is, a map $f$ to $B$,
defined on the variables of $\phi$,
that satisfies $\phi$ over $\relb$.
For each structure $\relb$, define
$\scsp(\relb)$ to be the problem of deciding, given 
a pair $(U, \phi)$ where $U$ is a set of variables and
$\phi$ is a $\wedge$-formula (over the signature of $\relb$)
with variables from $U$, whether or not there exists a 
surjective satisfying assignment on $U$, that is, a
surjective map
$f: U \to B$ 
that satisfies $\phi$ over $\relb$.\footnote{
We remark that, given an instance of a problem $\scsp(\relb)$,
variable equalities may be efficiently eliminated in a way 
that preserves the existence of a surjective satisfying
assignment~\cite[Proposition 2.1]{Chen14-hardness-surjective-csp}.
Likewise, given an instance of a problem $\csp(\relb)$,
variable equalities may be efficiently eliminated
in a way that preserves the existence of a satisfying
assignment.  Thus, in these problems,
whether or not one allows variable equalities
in instances is a matter of presentation, 
for the complexity issues at hand.}

Unless mentioned otherwise, when discussing NP-hardness
and NP-completeness, we refer to these notions
as defined with respect to polynomial-time many-one reductions.

\subsection{Definability and algebra}

Let $\relb$ be a relational structure.
Let $S$ be a finite set, and let $F$ be a set of mappings,
each of which is from $S$ to $B$.
We say that $F$
is \emph{$\wedge$-definable} over $\relb$ if there exists
a $\wedge$-formula $\phi$, whose variables are drawn from $S$,
such that $F$ is the set of satisfying assignments 
$f: S \to B$ of $\phi$, with respect to $\relb$.
Let $S$ be a finite set; when $T$ is a set of mappings, 
each of which is from $S$ to $B$, we use $\langle T \rangle_{\relb}$
to denote the smallest $\wedge$-definable 
set of mappings from $S$ to $B$ (over $\relb$) that contains $T$.
(Such a smallest set exists: over a structure $\relb$,
the set of all mappings from a finite set $S$ to $is$ B $\wedge$-definable via
the empty conjunction; and,
one clearly has $\wedge$-definability
of the intersection of two $\wedge$-definable sets
of mappings all sharing the same type.)

Let $T$ be a set of mappings from a finite set $I$ to 
a finite set $B$.
A \emph{partial polymorphism} of $T$ is a partial
mapping $p: B^J \to B$ such that, 
for any selection $s: J \to T$ of maps,
letting $c_i \in B^J$ denote the mapping taking each $j \in J$
to $(s(j))(i)$:
\begin{center}
if the mapping from $I$ to $B$ sending each $i \in I$ to
the value $p( c_i )$ is defined at each point $i \in I$,\\
then it is contained in $T$.
\end{center}
Conventionally, one speaks of a partial polymorphism 
of a relation $Q \subseteq B^k$; we here give a more
general formulation, as it will be convenient for us
to deal here with arbitrary index sets $I$.
When $Q \subseteq B^k$ is a relation,
we apply the just-given definition by viewing 
each element of $Q$ as a set of mappings
from the set $\{ 1, \ldots, k \}$ to $B$;
likewise, when $p:B^k \to B$ is a partial mapping
with $k \geq 0$ a natural number, we apply this definition
by viewing each element of $B^k$ as a mapping
from $\{ 1, \ldots, k \}$ to $B$.
A partial mapping $p: B^J \to B$ is a
partial polymorphism of a relational structure
if it is a partial polymorphism of each of the relations
of the structure.

The following is a known result; it connects $\wedge$-definability to
closure under partial polymorphisms.

\begin{theorem}
\label{thm:galois}
\cite{Romov81-algebras-partial-functions}
Let $\relb$ be a structure.
Let $T$ be a set of non-empty mappings 
from a finite set $I$ to $B$;
for each $i \in I$, let $\pi_i: T \to B$
denote the map defined by $\pi_i(t) = t(i)$.
The set $\langle T \rangle_{\relb}$
is equal to 
the set $T'$ of maps from $I$ to $B$
having a definition of the form $i \mapsto p(\pi_i)$,
where $p: B^T \to B$ is a partial polymorphism (of $\relb$)
with domain $\{ \pi_i ~|~ i \in I \}$.
\end{theorem}

We provide a proof of this theorem for the sake of completeness.

\begin{proof}
It is straightforward to verify that each
partial polymorphism of $\relb$ is also a partial
polymorphism of $\langle T \rangle_{\relb}$.
It thus holds that $T' \subseteq \langle T \rangle_{\relb}$.

In order to establish that
$\langle T \rangle_{\relb} \subseteq T'$, 
it suffices to show that $T'$ is $\wedge$-definable over $\relb$.
Since equalities are permitted in $\wedge$-formulas,
it suffices to prove the result in the case that the
maps $\pi_i$ are pairwise distinct.
Let $\theta$ be the $\wedge$-formula
where,
for each relation symbol $R$,
the atom $R(i_1, \ldots, i_k)$ is included as a conjunct
if and only if 
each $t \in T$ satisfies
$(t(i_1), \ldots, t(i_k)) \in R^{\relb}$.
By the definition of partial polymorphism,
we have that each map in $T'$ satisfies $\theta$.
On the other hand, when $u: I \to B$ is a map that satisfies
the formula $\theta$, 
consider the partial mapping
$q: B^T \to B$ sending $\pi_i$ to $u(i)$;
this is well-defined since the $\pi_i$ are pairwise distinct,
and is a partial polymorphism, by the construction of $\theta$.
\end{proof}

A \emph{polymorphism} is a partial polymorphism
that is a total mapping. 
A total mapping $p: B^J \to B$ is \emph{essentially unary}
if there exists $j \in J$ and a unary operation
$u: B \to B$ such that, for each mapping $h: J \to B$,
it holds that $p(h) = u(h(j))$.
An \emph{automorphism} of a structure $\relb$
is a bijection $\sigma: B \to B$ such that,
for each relation $R^{\relb}$ of $\relb$ and
for each tuple $(b_1, \ldots, b_k)$ whose arity $k$
is that of $R^{\relb}$,
it holds that $(b_1, \ldots, b_k) \in R^{\relb}$
if and only if $(\sigma(b_1), \ldots, \sigma(b_k)) \in R^{\relb}$.
It is well-known and straightforward to verify that,
for each finite structure $\relb$,
a bijection $\sigma: B \to B$ is an automorphism
if and only if it is a polymorphism.

Let $p: B^J \to B$ be a partial mapping.
For each $b \in B$, let $b^J$ denote the mapping
from $J$ to $B$
that sends each element $j \in J$ to $b$.
The \emph{diagonal} of $p$, denoted by $\hat{p}$,
is the partial unary mapping from $B$ to $B$
such that $\hat{p}(b) = p(b^J)$ for each $b \in B$.
With respect to a structure $\relb$,
we say that a partial mapping $p: B^J \to B$ is \emph{automorphism-like} 
when there exists $j \in J$ and an automorphism $\gamma$ (of $\relb$)
such that, for each mapping $h: J \to B$,
if $p(h)$ is defined, then it is equal to $\gamma(h(j))$.

\section{Framework}
\label{sect:framework}

Throughout this section, let $\relb$ be a 
finite relational structure, and let $B$ be its universe.
Let $I$ be a finite set; let $T$ be a set of mappings 
from $I$ to $B$.
Let us say that $T$ is \emph{surjectively closed} 
over $\relb$ if each surjective mapping in 
$\langle T \rangle_{\relb}$ is contained in 
$\{ \gamma \circ t ~|~ \textup{$\gamma$ is an automorphism of $\relb$, $t \in T$} \}$.
The following is essentially a consequence of Theorem~\ref{thm:galois}.

\begin{proposition}
\label{prop:surjectively-closed}
Let $I$, $T$ be as described.
For each $i \in I$, let $\pi_i: T \to B$ denote
the mapping defined by $\pi_i(t) = t(i)$.
The following are equivalent:

\begin{itemize}

\item The set $T$ is surjectively closed over $\relb$.
\item Each surjective partial polymorphism $p: B^T \to B$
(of $\relb$)
with domain $\{ \pi_i ~|~ i \in I \}$ is automorphism-like.

\end{itemize}
\end{proposition}

We consider a partial polymorphism $p: B^T \to B$ to be 
\emph{surjective} when, for each $b \in B$, there exists
$h \in B^T$ such that $p(h)$ is defined and is equal to $b$.

\begin{proof}
Suppose that $T$ is surjectively closed over $\relb$.
Let $p: B^T \to B$ be a surjective partial polymorphism
with the described domain.
By Theorem~\ref{thm:galois}, the map $t'$ defined by 
$i \mapsto p(\pi_i)$ is an element of $\langle T \rangle_{\relb}$;
note also that this map is surjective.
By the definition of surjectively closed,
it holds that $t'$ has the form $\gamma(t)$,
where $\gamma$ is an automorphism of $\relb$, and $t \in T$.
It follows that $p$ is automorphism-like.

Suppose that each surjective partial polymorphism
from $B^T$ to $B$ with the described domain is 
automorphism-like.  
By Theorem~\ref{thm:galois}, each surjective mapping $t'$ in
$\langle T \rangle_{\relb}$
is defined by $i \mapsto p(\pi_i)$ where $p$
is a partial polymorphism; note that $p$ is surjective.
By hypothesis, $p$ is automorphism-like;
it follows that $t'$ has the form 
$\gamma(t)$ where $\gamma$ is an automorphism of $\relb$,
and $t \in T$.  We conclude that $T$ is surjectively closed.
\end{proof}

In what follows, we generally use 
$D$ to denote a finite set,
$V$ to denote a finite set of variables, and
$G_{V,D}$ to denote the set of all mappings from $V$ to $D$; we will sometimes refer to elements of $G_{V,D}$ as \emph{assignments}.

Define an \emph{encoding} (for $\relb$) to be
a finite set $F$ of mappings, 
each of which is from a finite power $D^k$ of a finite set $D$ to $B$; we refer to $k$ as the
\emph{arity} of such a mapping.
Formally, an \emph{encoding} for $\relb$ is a finite set $F$ such that
there exists a finite set $D$ where, for each $f \in F$, there exists
$k \geq 0$ such that $f$ is a mapping from $D^k$ to $B$.
In what follows, we will give a sufficient condition for an encoding
to yield a reduction from a classical CSP over a structure
with universe $D$ to the problem $\scsp(\relb)$.
Assume $F$ to be an encoding;
define a \emph{$(V,F)$-application} to be a pair $(\tup{v}, f)$
consisting of a tuple of variables from $V$ and a mapping
$f \in F$ such that
the length of $\tup{v}$ is equal to the arity of $f$.
Let $A_{V,F}$ denote the set of all $(V,F)$-applications.

\begin{example}
Let $D = \{ 0, 1 \}$, $B = \{ 0, 1, 2 \}$,
and let $F$ be the encoding that 
contains the $3$ mappings from $D^0$ to $B$
as well as the $6$ injective mappings from $D^1$ to $B$.
We have $|F| = 9$.  
(This encoding will be used in Section~\ref{subsect:no-rainbow}.)

Let $V$ be a set of size $4$.  As $|V| = 4$ and $|D| = 2$,
we have $|G_{V,D}| = 2^4 = 16$.  
The arity $0$ mappings in $F$ give rise to $3$ $(V,F)$-applications
and the arity $1$ mappings in $F$ give rise to $6 |V| = 24$
$(V,F)$-applications.  Thus we have $|A_{V,F}| = 27$.
\end{example}

Let $V$ be a finite set of variables; let $F$ be an encoding.
For each $g \in G_{V,D}$, 
when $\alpha = ((v_1, \ldots, v_k), f)$ is an application
in $A_{V,F}$, 
define $\alpha[g]$ to be the value
$f(g(v_1), \ldots, g(v_k))$;
define $t[g]$ to be the map 
from $A_{V,F}$ to $B$ where
each application $\alpha$ is mapped to
$\alpha[g]$. Define $T_{V,F} = \{ t[g] ~|~ g \in G_{V,D} \}$.

\begin{proposition}
\label{prop:compute-closure}
Let $F$ be an encoding.
There 
exists a polynomial-time algorithm that, given a finite set $V$,
computes an $\wedge$-formula (over the signature of $\relb$)
defining $\langle T_{V,F} \rangle_{\relb}$.
\end{proposition}

Let us clarify how polynomial time is measured here:
we assume that the input to the algorithm is simply a finite set $V$,
and that such a set is encoded by a list where each element
appears explicitly; so, in particular, the length of a set's encoding
is larger than $|V|$.

\begin{proof}
The algorithm performs the following.  
For each relation  $R^{\relb}$ of $\relb$,
let $k$ be its arity.
For each tuple $(\alpha_1, \ldots, \alpha_k) \in A_{V,F}^k$,
the projection of $T$ onto $(\alpha_1, \ldots, \alpha_k)$
can be computed, by considering all possible assignments
on the variables in $V$ that appear in the $\alpha_i$;
note that the number of such variables is bounded by a constant,
since both the signature of $\relb$ and $F$ are assumed to be finite.
If this projection is a subset of $R^{\relb}$, then
include $R(\alpha_1, \ldots, \alpha_k)$ in the formula;
otherwise, do not.  Perform the same process for the
equality relation on $B$, with $k = 2$; whenever the projection
is a subset of this relation, include $\alpha_1 = \alpha_2$
in the formula.
\end{proof}

\begin{definition}
\label{def:stability}
Let $F$ be an encoding.
Define a relational structure $\relb$ to be
\emph{$F$-stable} if, for each non-empty finite set $V$,
it holds that
 each map in $T_{V,F}$ is surjective, and
 $T_{V,F}$ is surjectively closed (over $\relb$).
\end{definition}

Note that, 
relative to an encoding $F$,
 only the size of $V$ matters in the definition of $T_{V,F}$
in Definition~\ref{def:stability},
in the sense that when $V$ and $V'$ are of the same size,
$T_{V,F}$ and $T_{V',F}$ are equal up to relabelling of indices.

\begin{definition}
\label{def:induced-relation}
Let $F$ be an encoding.
An \emph{$F$-induced relation} of $\relb$ is 
a relation $Q' \subseteq D^s$ (with $s \geq 1$)
such that, letting $(u_1, \ldots, u_s)$ be 
a tuple of pairwise distinct variables,
there exists:
\begin{itemize}
\item 
a relation $Q \subseteq B^r$ that is
either a relation of $\relb$ or the equality relation on $B$,
and

\item 
a tuple $(\alpha_1, \ldots, \alpha_r) \in A^r_{U,F}$
\end{itemize}
such that
$Q' = \{ (g(u_1),\ldots,g(u_s)) ~|~ g \in G_{U,D},
((t[g])(\alpha_1), \ldots, (t[g])(\alpha_r)) \in Q \}$.
We refer to $Q$ as the relation that induces $Q'$,
and to the pair $(Q,(\alpha_1, \ldots, \alpha_r))$
as the \emph{definition} of $Q'$.
\end{definition}

\begin{definition}
Let $F$ be an encoding.
An \emph{$F$-induced template of $\relb$} is
a relational structure $\reld$ with universe $D$
and whose relations are all $F$-induced relations of $\relb$.
\end{definition}

\begin{theorem}
\label{thm:stable-gives-reduction}
Let $F$ be an encoding.
Suppose that $\relb$ is $F$-stable, and that
$\reld$ is an $F$-induced template of $\relb$.
Then, the problem $\csp(\reld)$ polynomial-time many-one
reduces to $\scsp(\relb)$.
\end{theorem}

\begin{proof}
Let $\phi$ be an instance of $\csp(\reld)$ 
with variables $V$.  
We may assume (up to polynomial-time computation) 
that $\phi$ does not include any variable equalities.
We create an instance $(A_{V,F},\psi)$ of
$\scsp(\relb)$;
this is done by computing two $\wedge$-formulas $\psi_0$ and $\psi_1$,
and setting $\psi = \psi_0 \wedge \psi_1$.

Compute $\psi_0$ to be an $\wedge$-formula
defining $\langle T_{V,F} \rangle_{\relb}$,
where $T_{V,F} = \{ t[g] ~|~ g \in G_{V,D} \}$;
such a formula is polynomial-time computable
by Proposition~\ref{prop:compute-closure}.

Compute $\psi_1$ as follows.
For each atom $R'(v_1, \ldots, v_s)$ of $\phi$
where $R'^{\reld}$ is induced by a relation $R^{\relb}$,
let $c: \{ u_1, \ldots, u_s \} \to V$ be the mapping
sending each $u_i$ to $v_i$, 
and include the atom $R(c(\alpha_1), \ldots, c(\alpha_r))$
in $\psi_1$; here,
$(\alpha_1, \ldots, \alpha_r)$ is the tuple from
Definition~\ref{def:induced-relation},
and
$c$ acts on an application $\alpha$ by being applied
individually
to each variable in the variable tuple of $\alpha$,
that is, when $\alpha = ((u'_1, \ldots, u'_k),f)$,
we have $c(\alpha) = ((c(u'_1), \ldots, c(u'_k)), f)$.
For each atom $R'(v_1, \ldots, v_s)$ of $\phi$ where
$R'^{\reld}$ is induced by the equality relation on $B$, 
let $c$ and $(\alpha_1, \alpha_2)$ be as above, and 
include the atom $c(\alpha_1) = c(\alpha_2)$ in $\psi_1$.
We make the observation that, 
from Definition~\ref{def:induced-relation},
a mapping $g \in G_{V,D}$ satisfies
an atom $R'(v_1, \ldots, v_s)$ of $\phi$
if and only if $t[g]$ satisfies the corresponding atom
or equality
in $\psi_1$.

We argue that $\phi$ is a \emph{yes} instance of
$\csp(\reld)$ if and only if $(A_{V,F},\psi)$ is a 
\emph{yes} instance
of $\scsp(\relb)$.
Suppose that $g \in G_{V,D}$ is a satisfying assignment
of $\phi$.  
The assignment $t[g]$ satisfies $\psi_0$ since
$t[g] \in T_{V,F}$.
Since the assignment $g$ satisfies each atom of $\phi$,
by the observation,
the assignment $t[g]$ satisfies each atom
and equality of $\psi_1$, and so
$t[g]$ is a satisfying assignment of $\psi_1$.
It also holds that $t[g]$ is surjective
by the definition of $F$-stable.
Thus, we have that $t[g]$ is a surjective satisfying
assignment of $\psi$.
Next, suppose that there exists 
a surjective satisfying assignment $t'$ of
$\psi$.  Since $t'$ satisfies $\psi_0$, it holds that
$t' \in \langle T_{V,F} \rangle_{\relb}$.
Since $T_{V,F}$ is surjectively closed (over $\relb$) by
$F$-stability, there exists $g \in G_{V,D}$
such that $t' = \gamma(t[g])$ for an automorphism
$\gamma$ of $\relb$.  Since $t'$ is a satisfying assignment
of $\psi$, so is $t[g]$; it then follows from
the observation that $g$ is a satisfying assignment of $\phi$.
\end{proof}

\newcommand{\sigmat}{\widetilde{\sigma}}

\subsection*{Inner symmetry}

Each set of the form $\langle T_{V,F} \rangle_{\relb}$
is closed under the automorphisms of $\relb$,
since each automorphism is a partial polymorphism
(recall Theorem~\ref{thm:galois}).
We here present another form of symmetry that
such a set $\langle T_{V,F} \rangle_{\relb}$
may possess, which we dub \emph{inner symmetry}.
Relative to a structure $\relb$ and an encoding $F$,
we define an \emph{inner symmetry} 
to be a pair $\sigma = (\rho,\tau)$
where $\rho: D \to D$ is a bijection, and
$\tau: B \to B$ is an automorphism of $\relb$ such that
$F = \{ \sigmat \circ f ~|~ f \in F \}$;
here, when $g: D^k \to B$ is a mapping,
$\sigmat(g): D^k \to B$ is defined as the composition
$\tau \circ g \circ (\rho,\ldots,\rho)$,
where $(\rho,\ldots,\rho)$ denotes the mapping
from $D^k$ to $D^k$ that applies $\rho$ to each entry
of a tuple in $D^k$.
When $\sigma$ is an inner symmetry,
we naturally extend the definition of $\sigmat$
so that it is defined on each application:
when $\alpha=(\tup{v}, f)$ is an application,
define $\sigmat(\alpha) = (\tup{v}, \sigmat(f))$.

The following theorem describes the symmetry on 
$\langle T_{V,F} \rangle_{\relb}$ induced by an inner symmetry.

\begin{theorem}
Let $\relb$ be a structure, let $F$ be an encoding,
and let $\sigma=(\rho,\tau)$ be an inner symmetry thereof.
Let $V$ be a non-empty finite set.
For any map $u: A_{V,F} \to B$, define $u': A_{V,F} \to B$
by $u'(\alpha) = u(\sigmat(\alpha))$;
it holds that $u \in \langle T_{V,F} \rangle_{\relb}$
if and only if $u' \in \langle T_{V,F} \rangle_{\relb}$.
\end{theorem}

\begin{proof}
For each $g \in G_{V,D}$,
define $t'[g]: A_{V,F} \to B$ to map each $\alpha \in A_{V,F}$
to $(\sigmat(\alpha))[g]$.
Define $T'_{V,F}$ as $\{ t'[g] ~|~ g \in G_{V,D} \}$.
By definition, 
$T'_{V,F} = \{ \tau(t[\rho(g)]) ~|~ g \in G_{V,D} \}$;
since $\rho$ is a bijection, we have that
$T'_{V,F} = \{ \tau(t[h]) ~|~ h \in G_{V,D} \}$.
Since $\tau$ is an automorphism of $\relb$,
we obtain 
$\langle T_{V,F} \rangle_{\relb} = \langle T'_{V,F} \rangle_{\relb}$.

Since $\sigma$ is an inner symmetry, we have
$F = \{ \sigmat \circ f ~|~ f \in F \}$,
from which it follows that the action of $\sigmat$
on applications in $A_{V,F}$ is a bijection on $A_{V,F}$.
We have $t'[g](\alpha) = t[g](\sigmat(\alpha))$.
For any map $u: A_{V,F} \to B$, define $u': A_{V,F} \to B$
by $u'(\alpha) = u(\sigmat(\alpha))$.
For all $u: A_{V,F} \to B$, we have 
$u' \in T'_{V,F} \Leftrightarrow u \in T_{V,F}$,
implying that 
$u' \in \langle T'_{V,F} \rangle_{\relb} 
\Leftrightarrow u \in \langle T_{V,F} \rangle_{\relb}$.
Since 
$\langle T_{V,F} \rangle_{\relb} = \langle T'_{V,F} \rangle_{\relb}$,
the theorem follows.
\end{proof}

\newcommand{\lf}{[}
\newcommand{\rf}{]}

\section{Hardness results}
\label{sect:hardness-results}

Throughout this section, we employ the following conventions.
When $B$ is a set and $b \in B$, we use the notation
$\tup{b}$ to denote the arity $0$ function
from $D^0$ to $B$ sending the empty tuple to $b$.
Let $V$ be a set, and let $F$ be an encoding.
When $\tup{b} \in F$, we 
overload the notation $\tup{b}$
and also use it
to denote the unique $(V,F)$-application in which it appears.
Relative to a structure $\relb$ (understood from the context),
when $\alpha_1, \ldots, \alpha_k \in A_{V,F}$ are applications
and $R$ is a relation symbol,
we write $R(\alpha_1,\ldots,\alpha_k)$ when,
for each $g \in G_{V,D}$, it holds that
$(\alpha_1[g], \ldots, \alpha_k[g] ) \in R^{\relb}$;
when $R$ is a symmetric binary relation, we also say that
$\alpha_1$ and $\alpha_2$ are adjacent.
When $\alpha$ is an application in $A_{V,F}$,
and $p: B^{T_{V,F}} \to B$ is a partial mapping,
we simplify notation by using $p(\alpha)$ to denote the value 
$p(\pi_{\alpha})$
(recall the definition of $\pi_{\alpha}$
from Proposition~\ref{prop:surjectively-closed}).

\subsection{Disconnected cut: the reflexive $4$-cycle}
\label{subsect:disconn-cut}

Let us use $\relc$ to denote the reflexive $4$-cycle,
that is, the structure with universe $C = \{ 0, 1, 2, 3 \}$
and single binary relation
$E^{\relc} = C^2 \setminus \{ (0,2),(2,0),(1,3),(3,1) \}$.
The problem  $\scsp(\relc)$ was shown to be NP-complete by
\cite{MartinPaulusma15-discut}; we here give a proof using our framework.
When discussing this structure, we will say that two values
$c, c' \in C$ are adjacent when $(c,c') \in E^{\relc}$.
Set $D = \{ 0, 1, 3 \}$. 
We use the notation $\lf a b c \rf$ to denote
the function $f: D \to C$ with $(f(0),f(1),f(3)) = (a,b,c)$,
so, for example $\lf 0 1 3 \rf$ denotes the identity mapping
from $D$ to $C$.
Define $F$ as the encoding
$$\{ \tup{0},\tup{1},\tup{2},\tup{3},
\lf 0 1 3 \rf,
\lf 0 1 0 \rf,
\lf 3 2 3 \rf,
\lf 3 1 3 \rf,
\lf 1 1 2 \rf,
\lf 0 0 3 \rf,
\lf 1 1 3 \rf
 \}.$$
We will prove the following.

\begin{theorem}
\label{thm:4-cycle}
The reflexive $4$-cycle $\relc$ is $F$-stable.  
\end{theorem}

We begin by observing the following.

\begin{proposition}
Define $\rho: D \to D$ as the bijection that swaps $1$ and $3$;
define $\tau: B \to B$ as the bijection that swaps $1$ and $3$.
The pair $\sigma=(\rho,\tau)$ is an inner symmetry of $F$
and $\relc$.
\end{proposition}

\begin{proof}
Consider the action of $\sigmat$ on $F$.
This action $\sigmat$ transposes $\tup{1}$ and $\tup{3}$;
$\lf 0 1 0 \rf$ and $\lf 0 0 3 \rf$;
$\lf 3 2 3 \rf$ and $\lf 1 1 2 \rf$; and,
$\lf 3 1 3 \rf$ and $\lf 1 1 3 \rf$.
It fixes each other element of $F$.
\end{proof}

\begin{proof} (Theorem~\ref{thm:4-cycle})
Let $V$ be a non-empty finite set; we need to show that
$T_{V,F}$ is surjectively closed.
We make use of Proposition~\ref{prop:surjectively-closed}.
Consider a partial polymorphism 
$p: C^{T_{V,F}} \to C$ whose domain is 
$\pi_\alpha$ over all $(V,F)$-applications $\alpha \in A_{V,F}$.

We first consider the situation where $p$ has a surjective diagonal;
we want to show that $p$ is automorphism-like.
Define $\beta: C \to C$ by $\beta(c) = p(\tup{c})$.
We have that $\beta$ is surjective and
a polymorphism of $\relc$, implying that it is an automorphism of $\relc$.
By considerations of symmetry (in particular, by replacing $p$
with the translation $\beta^{-1} \circ p$, and then
translating back post-argument), we may 
assume that 
$(p(\tup{0}),p(\tup{1}), p(\tup{2}),p(\tup{3})) = (0,1,2,3)$.
Under this assumption, it is straightforward to verify the
following fact:
for any application $(v,[abc])$, it holds that
$p(v,[abc]) \in \{ a, b, c \}$.
For example, in the case that the set
$\{ a, b, c \}$ has size $2$ and its two elements $a', b'$
are adjacent in $\relc$,
$(v,[abc])$ is adjacent to $\tup{a'}$ and $\tup{b'}$, 
implying that $p(v,[abc])$ is adjacent to both
$p(\tup{a'}) = a'$ and $p(\tup{b'}) = b'$, and is thus in $\{ a', b' \}$.
In the following reasoning, we implicitly use this fact, continually.

Fix $v$ to be a variable.  We show that 
there exists a map $g_v: \{ v \} \to D$ such that,
for each $f \in F$, $p(v,f) = f(g_v(v))$.  This suffices,
since then one can define a map $g \in G_{V,D}$ that extends all
of the maps $(g_v)$ to derive the map $t[g] \in T_{V,F}$ equal to
the map sending each application $\alpha$ to $p(\alpha) = p(\pi_\alpha)$.

We have that $(v,[010])$, $(v,[323])$, and $(v,[313])$
are pairwise adjacent.
When the value of $p(v,[010])$ is $0$,
we have, by the identified adjacencies,
$p(v,[323]) = 3$ and $p(v,[313]) = 3$, which by the adjacency
of $(v,[313])$ and $(v,[003])$, implies $p(v,[003])$ is $0$ or $3$.
When the value of $p(v,[010])$ is $1$, 
we have, by the identified adjacencies,
$p(v,[323]) = 2$ and $p(v,[313]) = 1$, which by the adjacency
of $(v,[313])$ and $(v,[003])$, implies $p(v,[003])$ is $0$.
We thus have $3$ cases, depending on the value of the tuple
$(p(v,[010]), p(v,[323]), p(v,[313]), p(v,[003]))$:
this tuple is equal to either
$(0,3,3,0)$, $(1,2,1,0)$, or $(0,3,3,3)$.
By using the pairwise adjacency of $(v,[112])$, $(v,[003])$, and
$(v,[113])$, we can confirm that in the $3$ cases, these 
applications are mapped by $p$ to $(1,0,1)$, $(1,0,1)$, and $(2,3,3)$,
respectively.
Then, by using the adjacency of $(v,[013])$ with each of
$(v,[010])$, $(v,[323])$, and $(v,[112])$, we can confirm that
$p(v,[013])$ is, in the $3$ cases, equal to $0$, $1$, and $3$, respectively.

We have analyzed the situation where $p$ has a surjective diagonal.
To establish the result, it suffices to
show that if $p$ has a non-surjective diagonal,
then it is not surjective; this is what we do in the rest of the proof.

By considerations of symmetry 
(namely, by the automorphisms and by the inner symmetries), 
it suffices to consider
the following values for 
the diagonal values $(p(\tup{0}),p(\tup{1}), p(\tup{2}),p(\tup{3}))$:
$(0,0,0,0)$, $(0,0,1,0)$, $(0,1,0,0)$, $(1,0,0,0)$,
$(0,0,1,1)$, $(0,1,0,1)$, $(0,1,2,1)$, $(0,1,0,3)$.
We consider each of these cases.

In each of the first $3$ cases, we argue as follows.  
Consider an application $(v,f)$ with 
$f: D \to B$ in $F$ and $v \in V$;
it is adjacent to $\tup{0}$,
adjacent to $\tup{3}$, or 
adjacent to both $\tup{1}$ and $\tup{2}$; thus,
for such an application, we have $p(v,f)$ is adjacent to $0$.
It follows that for no such application do we have $p(v,f) = 2$,
and so $p$ is not surjective.

Case: diagonal  $(1,0,0,0)$.
Observe that any application $(v,f)$ with 
$f: D \to C$ in $F \setminus \{ \lf 0 1 3 \rf \}$ and $v \in V$;
is adjacent to 
$\tup{1}$, $\tup{2}$, or $\tup{3}$,
and thus, for any such application, $p(v,f)$ is adjacent to $0$.
Thus if $2$ is in the image of $p$,
there exists a variable $u \in V$ such that
$p(u,\lf 0 1 3 \rf) = 2$.  But then, for each
$f: D \to C$ in $F \setminus \{ \lf 0 1 3 \rf \}$,
we have that $(u,\lf 0 1 3 \rf)$ and $(u,f)$ are adjacent.
$(u,\lf 0 1 0 \rf)$ and $(u,\lf 0 0 3 \rf)$
are adjacent to $\tup{0}$, which $p$ maps to $1$, 
and to $(u,\lf 0 1 3 \rf)$, which $p$ maps to $2$;
this, by the observation,
$p(u,\lf 0 1 0 \rf) = p(u, \lf 0 0 3 \rf) = 1$.
$(u,\lf 3 2 3 \rf)$ and $(u,\lf 1 1 2 \rf)$ are adjacent
to the just-mentioned applications, from which we obtain
$p(u,\lf 3 2 3 \rf) = p(u,\lf 1 1 2 \rf) = 1$.

In order for $p$ to be surjective,  there exists
a different variable $u' \in V$ and $f' \in F$ such that
$p(u',f') = 3$. 
It must be that $(u',f')$ is not adjacent to $\tup{0}$.
But then $f'$ is $\lf 3 2 3 \rf$ or $\lf 1 1 2 \rf$,
and this contradicts that $(u',f')$ is adjacent to
$(u,f)$.

Case: diagonal $(0,0,1,1)$.
There must be an application $(v,f)$ with 
$p(v,f) = 2$.  Since $(v,f)$ cannot be adjacent to 
$\tup{0}$ nor $\tup{1}$, it must be that
$f = \lf 3 2 3 \rf$.  There must also be an application
$(v', f')$ with $p(v',f') = 3$; this application
cannot be adjacent to $\tup{2}$ nor $\tup{3}$, 
and so $f'$ is $\lf 0 1 3 \rf$ or $\lf 0 1 0 \rf$.
$(v',\lf 3 2 3 \rf)$ is adjacent to $(v',f')$,  to
 $\tup{3}$, and to $(v,\lf 3 2 3 \rf)$, and so
$p(v',\lf 3 2 3 \rf)$ is $2$.  
$(v', \lf 3 1 3 \rf)$ is adjacent to $(v',\lf 3 2 3 \rf)$,
to $f'$, to $\tup{0}$, and to $\tup{2}$, and thus it cannot
be mapped to any value.

Case: diagonal $(0,1,0,1)$ or $(0,1,0,3)$.
For each variable $u \in F$ and each $f: D \to B$
in $F$, it holds that $(u,f)$ is adjacent to
either $\tup{0}$ or $\tup{2}$.
It follows that no such pair $(u,f)$ maps to $2$ under $p$.

Case: diagonal $(0,1,2,1)$.
If $p$ is surjective, there exists
$(v,f)$ such that $p(v,f) = 2$.  
$(v,f)$ cannot be adjacent to $\tup{0}$, implying
that $f = \lf 3 2 3 \rf$ or $\lf 1 1 2 \rf$.
When $f = \lf 3 2 3 \rf$, we infer that 
$p(v,\lf 0 1 0 \rf) = 1$, and then that
$p(v,\lf 0 1 3 \rf) = p(v,\lf 3 1 3 \rf) = p(v,\lf 1 1 3 \rf) = 1$.
Similarly, when $f = \lf 1 1 2 \rf$,
we infer that
$p(v,\lf 0 0 3 \rf) = 1$, and then that
$p(v,\lf 0 1 3 \rf) = p(v,\lf 3 1 3 \rf) = p(v,\lf 1 1 3 \rf) = 1$.
There exists $(v',f')$ such that $p(v',f') = 3$.
$(v',f')$ cannot be adjacent to $\tup{1}$ nor $\tup{3}$, 
so $f'$ is one of $\lf 0 1 3 \rf$, $\lf 3 1 3 \rf$, $\lf 1 1 3 \rf$.
This contradicts that $(v',f')$ and $(v,f')$ are adjacent.
\end{proof}

\begin{theorem}
The problem $\scsp(\relc)$ is NP-complete.
\end{theorem}

\begin{proof}
Define $\reld'$ (following~\cite[Section 2]{MartinPaulusma15-discut})
to be the structure with universe $\{ 0, 1, 3 \}$
and with relations
$$S_1^{\reld'} = \{ (0,3),(1,1),(3,1),(3,3) \},
S_2^{\reld'} = \{ (1,0),(1,1),(3,1),(3,3) \},$$
$$S_3^{\reld'} = \{ (1,3),(3,1),(3,3) \},
S_4^{\reld'} = \{ (1,1),(1,3),(3,1) \}.$$
Each of these relations is the intersection of
binary $F$-induced relations:
for
$S_1^{\reld'}$, use the definitions
$(E^{\relc}, (\tup{2}, (u_2,\lf 0 1 3 \rf)))$,
$(E^{\relc}, ((u_1, \lf 0 1 3 \rf), (u_2,\lf 3 2 3 \rf)))$;
for $S_2^{\reld'}$, 
the definitions
$(E^{\relc}, (\tup{2}, (u_1,\lf 0 1 3 \rf)))$,
$(E^{\relc}, ((u_1, \lf 1 1 2 \rf), (u_2,\lf 0 1 3 \rf)))$;
for $S_3^{\reld'}$, the definitions 
$(E^{\relc}, (\tup{2}, (u_1,\lf 0 1 3 \rf)))$,
$(E^{\relc}, (\tup{2}, (u_2,\lf 0 1 3 \rf)))$,
$(E^{\relc}, ((u_1, \lf 0 0 3 \rf), (u_2,\lf 3 2 3 \rf)))$;
and, for $S_4^{\reld'}$, the definitions
$(E^{\relc}, (\tup{2}, (u_1,\lf 0 1 3 \rf)))$,
$(E^{\relc}, (\tup{2}, (u_2,\lf 0 1 3 \rf)))$,
$(E^{\relc}, ((u_1, \lf 1 1 2 \rf), (u_2,\lf 0 1 0 \rf)))$.
Let $\reld$ be the $F$-induced template
of $\relb$ whose relations are all of the mentioned
$F$-induced relations.
Then, we have $\csp(\reld')$ reduces to $\csp(\reld)$,
and that $\csp(\reld)$ reduces to $\scsp(\relc)$ by 
Theorem~\ref{thm:stable-gives-reduction}.
The problem $\csp(\reld')$ is NP-complete, as argued in
\cite[Section 2]{MartinPaulusma15-discut}, 
and thus we conclude that $\scsp(\relc)$ is NP-complete.
\end{proof}

\subsection{No-rainbow $3$-coloring}
\label{subsect:no-rainbow}

Let $\reln$ be the structure with universe $N = \{ 0, 1, 2 \}$
and a single ternary relation
$$R^{\reln} = \{ (a,b,c) \in N^3 ~|~ \{ a, b, c \} \neq N \}.$$
The problem $\scsp(\reln)$ was first shown to be NP-complete
by Zhuk~\cite{Zhuk20-norainbow}; we give a proof 
which is akin to proofs given by 
Zhuk~\cite{Zhuk20-norainbow},
using our framework.

Define $D = \{ 0, 1 \}$, and define $F$ as
$\{ \tup{0}, \tup{1}, \tup{2} \} \cup U$
where $U$ is the set of all injective mappings from $D$ to $N$.
We use the notation $\lf a b \rf$
to denote the mapping $f: D \to N$ with $(f(0),f(1)) = (a,b)$.
Let $\iota_D: D \to D$ denote the identity mapping on $D$;
it is straightforwardly verified that, for each bijection
$\tau: N \to N$, the pair $(\iota_D, \tau)$ is 
an inner symmetry of $\reln$ and $F$.

\begin{theorem}
The structure $\reln$ is $F$-stable.	
\end{theorem}

\begin{proof}
Let $V$ be a non-empty finite set; we need to show that
$T_{V,F}$ is surjectively closed.
We use Proposition~\ref{thm:galois}.
It is straightforward to verify that each surjective partial polymorphism
of the described form having a surjective diagonal is 
automorphism-like.
Consider a partial polymorphism $p: N^{T_{V,F}} \to N$
whose domain is $\pi_\alpha$ over all
$(V,F)$-applications $\alpha \in A_{V,F}$.
We show that if $p$ has a non-surjective diagonal, 
then it is not surjective.
By considerations of symmetry
(namely, by the automorphisms and by the inner symmetries),
we need only consider
the following values 
for the diagonal $(\hat{p}(0),\hat{p}(1),\hat{p}(2))$:
$(0,0,0)$, $(0,1,1)$.

Diagonal $(0,0,0)$. Assume $p$ is surjective;
there exist applications 
$(v,\lf a b \rf)$,
$(v',\lf a' b' \rf)$ 
such that
$p(v,\lf a b \rf) = 1$,
$p(v',\lf a' b' \rf) = 2$.
In the case that $\{ a, b \} = \{ a', b' \}$,
we have $R( \tup{a}, (v,\lf a b \rf), (v', \lf a' b \rf))$
but that these applications are, under $p$,
equal to $(0,1,2)$, a contradiction.
Otherwise, there is one value in $\{a, b \} \cap \{ a', b \}$;
suppose this value is $b = b'$.
Let $c$ be the value in $N \setminus \{ a, b \}$,
and $c'$ be the value in $N \setminus \{ a', b' \}$.
We claim that $p(v,\lf a c \rf) = 1$: 
if it is $2$, we get a contradiction via
$R(\tup{b},(v,\lf a b \rf), (v,\lf a c \rf))$,
and if it is $0$, we get a contradiction via
$R((v, \lf a c \rf), (v, \lf a b \rf), (v', \lf a' b' \rf))$.
By analogous reasoning, we obtain that
$p(v', \lf a c' \rf) = 2$.
But since $\{ a, c \} = \{ a', c' \}$, we may reason
as in the previous case to obtain a contradiction.

Diagonal $(0,1,1)$.  Assume $p$ is surjective;
there exists an application $(v,\lf a b \rf)$
such that $p(v,\lf a b \rf) = 2$.
We have that $\{ a, b \} \neq \{ 0, 1 \}$,
for if not, we would have a contradiction via
$R(\tup{0},\tup{1},(v,\lf a b \rf))$.
Analogously, we have that $\{ a, b \} \neq \{ 0, 2 \}$.
Thus, we have $\{ a, b \} = \{ 1, 2 \}$.
Suppose that $\lf a b \rf$ is $\lf 1 2 \rf$
(the case where it is $\lf 2 1 \rf$ is analogous).
Consider the value of $p(v,\lf 0 2 \rf)$:
if it is $0$,
we have a contradiction via 
$R((v,\lf 0 2 \rf),\tup{1},(v,\lf 1 2 \rf))$,
if it is $1$,
we have a contradiction via
$R(\tup{0},(v,\lf 0 2 \rf),(v,\lf 1 2 \rf))$;
if it is $2$,
we have a contradiction via 
$R(\tup{0},\tup{2},(v,\lf 0 2 \rf))$.
\end{proof}

\begin{theorem}
\label{thm:no-rainbow-np-complete}
The problem $\scsp(\reln)$ is NP-complete.	
\end{theorem}

\begin{proof}
The \emph{not-all-equal} relation
$\{ 0, 1 \}^3 \setminus \{ (0,0,0),(1,1,1) \}$
is an $F$-induced relation of $\reln$,
via the definition 
$(R^{\reln}, (u_1,\lf 0 1 \rf), (u_2,\lf 1 2 \rf), (u_3,\lf 2 0 \rf))$.
It is well-known that the problem $\csp(\cdot)$
on a structure having this relation is NP-complete
via Schaefer's theorem,
and thus we obtain the result by Theorem~\ref{thm:stable-gives-reduction}.
\end{proof}

\subsection{Diagonal-cautious clones}
\label{subsect:diagonal-cautious}

We show how the notion of stability can be used to derive the previous hardness result of the present author~\cite{Chen14-hardness-surjective-csp}.
Let $B = \{ b_1^*, \ldots, b^*_n \}$ be a set of size $n$.
When $\relb$ is a structure with universe $B$,
we use $\relb^*$ to denote the structure obtained from $\relb$
by adding, for each $b_i^* \in B$, a relation 
$\{ (b_i^*) \}$.
A set $C$ of operations on $B$ is \emph{diagonal-cautious} if there exists a map 
$G: B^n \to \wp(B)$ such that:

\begin{itemize}
\item for each operation $f \in C$,
it holds that $\image(f) \subseteq G(\hat{f}(b_1^*), \ldots, \hat{f}(b_n^*))$, and
\item for each tuple $(b_1, \ldots, b_n) \in B^n$,
if $\{ b_1, \ldots, b_n \} \neq B$, then
$G(b_1, \ldots, b_n) \neq B$.
\end{itemize}

The previous hardness result~\cite{Chen14-hardness-surjective-csp} that we rederive here is 
that when $\relb$ is a structure whose polymorphisms are \emph{diagonal-cautious},
it holds that $\csp(\relb^*)$ reduces to
$\scsp(\relb)$.

\begin{theorem}
Suppose that the set of polymorphisms 
of a relational structure $\relb$ is diagonal-cautious, 
and that the universe $B$ of $\relb$ has size $n \geq 2$.  There exists an encoding $F$, whose elements each have arity $\leq 1$,
such that:
\begin{itemize}
\item
the structure $\relb$ is $F$-stable,
and 
\item there is a surjective mapping $f_x: D \to B$ in $F$
where, for each relation $Q \subseteq B^k$ of $\relb$,
the relation 
$\bigcup_{(b_1,\ldots,b_k) \in Q} (f_x^{-1}(b_1) \times \cdots \times f_x^{-1}(b_k))$ is an $F$-induced relation.
\end{itemize}
It consequently holds that $\csp(\relb^*)$ reduces to $\scsp(\relb)$.
\end{theorem}

\begin{proof}
Suppose that the polymorphisms of $\relb$ are diagonal-cautious
via $G: B^n \to \wp(B)$.
By Lemma 3.3 of~\cite{Chen14-hardness-surjective-csp},
there exists a relation $P \subseteq B^{(n^n)}$ %
with the following properties:
\begin{enumerate}

\item[(0)] $P$ is $\wedge$-definable.

\item[(1)] For each tuple $(b_1,\ldots,b_n,c,d_1,\ldots,d_m) \in P$,
it holds that each entry of this tuple is in 
$G(b_1, \ldots, b_n)$.

\item[(2)] For each $c \in B$, there exist values
$d_1,\ldots,d_m \in B$
such that $(b_1^*,\ldots,b_n^*,c,d_1,\ldots,d_m) \in P$.

\item[(3)] For each tuple $(b_1,\ldots,b_n,c,d_1,\ldots,d_m) \in P$,
there exists a polymorphism $p^+: B \to B$ of $\relb$
such that $p^+(b_i^*) = b_i$.

\end{enumerate}

We associate the coordinates
of $P$ with the variables $(v_1,\ldots,v_n,x,y_1,\ldots,y_m)$.
In the scope of this proof, when $z$ is one of these variables,
 we use $\pi_z$ to denote the operator that projects a tuple
 onto the coordinate corresponding to $z$.
Let $P'$ be the subset of $P$ that
contains each tuple $q \in P$ such that, 
for each $i = 1, \ldots, n$, it holds that
$\pi_{v_i}(q) = b_i^*$.
Let $q_1, \ldots, q_\ell$ be a listing of the tuples in $P'$.
Let $D = \{ 1, \ldots, \ell \}$, 
and let $F$ contain, for each 
$z \in \{ v_1, \ldots, v_n, x, y_1, \ldots, y_m \}$,
the map 
$f_z: D \to B$ defined by
$f_z(i)= \pi_z(q_i)$.
Observe that $f_x$ is surjective, by (2).

We verify that the structure $\relb$ is $F$-stable, as follows.
Let $V$ be a finite non-empty set.
Observe that for each $g \in G_{V,D}$ and each $u \in V$,
the map $t[g]$ sends the applications
$(u,f_{v_1}), \ldots, (u,f_{v_n})$
to $b_1^*, \ldots, b_n^*$, respectively; thus,
each map in $T_{V,F}$ is surjective.
To show that $T_{V,F}$ is surjectively closed, 
consider a mapping $t' \in \langle T_{V,F} \rangle_{\relb}$.
Observe first that because for any $u_1, u_2 \in V$
(and any $i$),
$(u_1, f_{v_i})$ and $(u_2,f_{v_i})$ are sent to the same value
by any map in $T_{V,F}$, the same holds for any map in 
$\langle T_{V,F} \rangle_{\relb}$, and for $t'$ in particular.
Let $u \in V$; since $P$ is $\wedge$-definable
and the tuple of applications
$$s^! = ((u,f_{v_1}), \ldots, (u,f_{v_n}), (u,f_x),
(u,f_{y_1}),\ldots, (u,f_{y_m}))$$
is pointwise mapped by each $t \in T_{V,F}$ to a tuple in $P'$,
it holds that this tuple is sent by $t'$ to a tuple in $P$.
\begin{itemize}
\item
Case:
Suppose that the restriction of $t'$ to 
$(u,f_{v_1}), \ldots, (u,f_{v_n})$ is not surjective onto $B$.
Then by (1), the restriction of $t'$ to
$(u,f_{v_1}), \ldots, (u,f_{v_n}), (u,f_x),
(u,f_{y_1}),\ldots, (u,f_{y_m})$ has image contained in the set
$G(t'(u,f_{v_1}), \ldots, t'(u,f_{v_n}))$,
which is not equal to $B$ by the definition of diagonal-cautious.
Since $u$ was chosen arbitrarily, we obtain that $t'$ is
not surjective.

\item Case: Suppose that the previous case's assumption does not hold.
We have that the tuple $s^!$ is pointwise mapped by $t'$ to a tuple 
$(b_1, \ldots, b_n, c, d_1, \ldots, d_n)$
in $P$.  Thus, by (3),
there exists a unary polymorphism $p^+$
such that $p^+(b_i^*) = b_i = t'(u,f_{v_i})$, for each $i$.
By this case's assumption, we have $\{ b_1, \ldots, b_n \} = B$; thus,
the operation $p^+$ is a bijection, and hence an automorphism
of $\relb$.  
As $p^+$ and each of its powers is thus a partial polymorphism
of $P$,
it follows that $s^!$ mapped under  $(p^+)^{-1}(t')$ 
is inside $P$ and,
due to the choice of $p^+$, inside $P'$.
Since $u$ was chosen arbitrarily, we obtain that 
$(p^+)^{-1}(t')$ is in $T_{V,F}$, and so $t'$ has the desired form.
\end{itemize}

For any relation $Q \subseteq B^k$ of $\relb$,
consider the $F$-induced relation $Q' \subseteq D^k$ of $\relb$ 
defined by $(Q, (u_1,f_x),\ldots,(u_k,f_x))$.
We have that $Q'$ is the union of
$f_x^{-1}(b_1) \times \cdots \times f_x^{-1}(b_k)$
over all tuples $(b_1, \ldots, b_k) \in Q$.
Moreover, for each $b \in B$, each set $f_x^{-1}(b)$
is an $F$-induced relation.
We may conclude that $\relb^*$, the expansion of
$\relb$ by all constant relations 
(those relations $\{ (b) \}$, over $b \in B$),
has that $\csp(\relb^*)$ reduces to $\scsp(\relb)$.
\end{proof}

As discussed in the previous article~\cite[Proof of Corollary 3.5]{Chen14-hardness-surjective-csp},
when a structure $\relb$ (with non-trivial universe size) has 
only essentially unary polymorphisms, it holds that
the polymorphisms of $\relb$ are diagonal-cautious,
and that the problem $\csp(\relb^*)$ is NP-complete.
With these facts in hand, we obtain the following corollary of the just-given theorem.

\begin{corollary}
\label{cor:eu}
Suppose that $\relb$ is a finite structure whose universe $B$
has size strictly greater than $1$.  
If each polymorphism of $\relb$ is essentially unary,
the problem $\csp(\relb^*)$, and hence the problem $\scsp(\relb)$,
is NP-complete.
\end{corollary}

It is known, under the assumption that P does not equal NP,
that the problem $\scsp(\relb)$ for any $2$-element structure $\relb$
is NP-complete if and only if each polymorphism 
of $\relb$ is essentially unary (this follows from~\cite[Theorem 6.12]{CreignouKhannaSudan01-boolean}).
Hence, under this assumption,
the hardness result of Corollary~\ref{cor:eu}
covers all hardness results for the problems $\scsp(\relb)$
over each $2$-element structure $\relb$.

In the remainder of this section, we show that
the previous hardness result~\cite{Chen14-hardness-surjective-csp} rederived here provably does not apply
to the structures $\relc$ and $\reln$ of 
Sections~\ref{subsect:disconn-cut}
and~\ref{subsect:no-rainbow}, respectively.

\begin{theorem}
The set of polymorphisms of the structure $\relc$	
is not diagonal-cautious.
\end{theorem}

\begin{proof}
Suppose that the stated set is diagonal-cautious
via $G: C^4 \to \wp(C)$.
We consider polymorphisms 
$f: C^2 \to C$ of $\relc$
such that the tuple of diagonal values
$(\hat{f}(0),\hat{f}(1),\hat{f}(2),\hat{f}(3))$
is equal to $(0,1,0,1)$.
By the first condition in the definition
of diagonal-cautious,
it must hold for such a polymorphism that
$\image(f) \subseteq G(0,1,0,1)$.  

There exists such a polymorphism having
$2$ in its image,
namely, the operation
$g_2$ defined as follows:
$g_2(a,b)$ is $2$ when $(a,b) = (0,2)$;
$b \mod 2$ when $a=2$ or $b=0$; and,
$1$ otherwise.
Analogously, 
there exists such a polymorphism having
$3$ in its image,
namely, the operation
$g_3$ defined as follows:
$g_3(a,b)$ is $3$ when $(a,b) = (1,3)$;
$b \mod 2$ when $a=3$ or $b=1$; and,
$0$ otherwise.

It follows that $\{ 0, 1, 2, 3 \} = G(0,1,0,1)$,
which contradicts the second condition
in the definition of diagonal-cautious.
\end{proof}

\begin{theorem}
The set of polymorphisms of the structure $\reln$	
is not diagonal-cautious.
\end{theorem}

\begin{proof}
Suppose that the stated set is diagonal-cautious
via $G: N^3 \to \wp(N)$.
It is straightforward to verify that each
non-surjective operation
$f: N^2 \to N$ is a polymorphism of $\reln$
(indeed, each non-surjective operation on $N$
is a polymorphism of $\reln$).
Consider the non-surjective operations
$f: N^2 \to N$ having diagonal $0$, 
that is, such that 
$\hat{f}(0) = \hat{f}(1) = \hat{f}(2) = 0$.
For each such operation, it holds
that $\image(f) \subseteq G(0,0,0)$, by
the first condition
in the definition of diagonal-cautious.
But since both $1$ and $2$ fall into the images
of such operations, it holds that
$G(0,0,0) = \{ 0,1,2 \}$.  
This contradicts the second condition
in the definition of diagonal-cautious.
\end{proof}

\section{Establishing stability}
\label{sect:decidability}
In this section,
we present a decidability result
for stability in the case that $F$ contains only
maps of arity at most $1$.
For each $b \in B$, we use $b_0$ to denote the mapping
from $D^0$ to $B$ that sends the empty tuple to $b$;
we use $B_0$ to denote $\{ b_0 ~|~ b \in B \}$.

\begin{theorem}
\label{thm:bound-num-vars}
Let $\relb$, $D$, and $F$ be as described.
Suppose that $F$ contains only maps of arity $\leq 1$,
and that it
contains $B_0$ as a subset.
The structure $\relb$ is $F$-stable if and only if
when $W$ has size $\leq |B|$, the set
$T_{W,F} = \{ t[g] ~|~ g \in G_{W,D} \}$ is surjectively closed
(over $\relb$).
\end{theorem}

\begin{proof}
The forward direction is immediate, so we prove the backward
direction.
We need to show that, for each non-empty finite set $V$,
the set $T_{V,F} = \{ t[g] ~|~ g \in G_{V,D} \}$ fulfills
the conditions given in Definition~\ref{def:stability}.
Each map in $T_{V,F}$ is surjective since $B_0 \subseteq F$.
We thus argue that $T_{V,F}$ is surjectively closed.
This follows from the assumptions when $|V| \leq |B|$.

When $|V| > |B|$, suppose
(for a contradiction) that $T_{V,F}$ is not surjectively
closed.  There exists a surjective mapping $t'$
in $\langle T_{V,F} \rangle_\relb$ violating the
definition of surjectively closed.
For any variable $v \in V$, 
let $A_v$ be the set of applications that involve either $B_0$
or the variable $v$.  Let us say that $v$ is \emph{well-formed}
if the map $t'$ restricted to $A_v$ has the form
$\beta_v(t[h_v])$ where $\beta_v$ is an automorphism (of $\relb$)
and $h_v$ is a map in $G_{\{v\}, D}$.
We claim that there exists a variable in $V$ that is not well-formed.
Suppose, for a contradiction,
 that each variable $v \in V$ is well-formed;
then, the automorphisms $\beta_v$ over $v \in V$ are are equal
to each other, since each set $A_v$ of applications includes
each application $\alpha_b$
defined as $(\epsilon,b_0)$ for each $b \in B$
(where $\epsilon$ is the empty tuple), and
this application is is mapped by $t'$
to the value $t'(\alpha_b)$ which 
by well-formedness is equal to $\beta_v(\alpha_b[h_v])$
(as $\alpha_b$ begins with the empty tuple, the value $\alpha_b[h_v]$
does not depend on $h_v$).
Let $\beta$ denote the common value of the automorphisms $\beta_v$
over $v \in V$.
Then, the tuple $t'$ is equal to $\beta(t[h])$, where $h \in G_{V,D}$
is the unique map extending each of the maps $(h_v)_{v \in V}$;
this implies that $t'$ does not violate the definition of
surjectively closed (for $T_{V,F}$), contradicting the choice of the map $t'$.
The claim follows: there exists a variable $u' \in V$ that
is not well-formed.

Since $u'$ is not well-formed,
the map $t'$ restricted to $A_{u'}$ 
does not have the form $\beta'(t[h'])$
for an automorphism $\beta'$ and a map $h' \in G_{\{u'\},D}$.
As $B_0 \subseteq F$, it is possible to pick
a set $A^-$ of applications using 
at most $|B|-1$ variables such that
the restriction of $t'$ to $A^-$ is surjective.
Let $U$ be the set of variables appearing in $A^-$;
we have $|U| \leq |B|-1$.
Set $U' = U \cup \{ u' \}$.

It is an immediate consequence of Theorem~\ref{thm:galois}
that the restriction of $t'$ to $A_{U', F}$
is contained in $\langle T_{U',F} \rangle_{\relb}$,
where $T_{U',F} = \{ t[g] ~|~ g \in G_{U',D} \}$.
By the choice of $u'$, we have that
the described restriction of $t'$ 
is not in $\langle T_{U',F} \rangle_{\relb}$;
since this restriction, by the choice of $U$, is surjective,
we obtain that $T_{U',F}$ is not surjectively closed,
a contradiction to the assumption of the backward direction.
\end{proof}

\begin{theorem}
There exists an algorithm that, given a structure $\relb$
and a set $F$ of maps each having arity $\leq 1$ and
where $B_0 \subseteq F$, decides whether or not
$\relb$ is $F$-stable.
\end{theorem}

\begin{proof}
The algorithm checks the condition of Theorem~\ref{thm:bound-num-vars},
in particular, it checks whether or not, for all $i = 1, \ldots, |B|$,
the set $T_{W,F}$ is surjectively
closed for a set $W$ having size $i$.  
\end{proof}

\section{Condensations}
\label{sect:condensations}

Define 
a \emph{condensation} from a relational structure
$\rela$ to a relational structure $\relb$ to be a
homomorphism $h: A \to B$ from $\rela$ to $\relb$
such that $h(A) = B$ and $h(R^{\rela}) = R^{\relb}$
for each relation symbol $R$.  That is, a condensation
is a homomorphism that maps the universe
of the first structure surjectively onto the universe of the second,
and in addition, maps each relation of the first structure
surjectively onto the corresponding relation of the second structure.
When $\relb$ is a structure,
we define the problem $\cond(\relb)$ of deciding whether or not
a given structure admits a condensation to $\relb$,
using the formulation of the present paper,
as follows.
The instances of $\cond(\relb)$ are the instances
of the problem $\scsp(\relb)$;
an instance $(U, \phi)$
is a \emph{yes} instance
if there exists a surjective map $f: U \to B$
that satisfies $\phi$ (over $\relb$) 
such that, for each relation symbol $R$
and each tuple $(b_1, \ldots, b_k) \in R^{\relb}$,
there exists an atom $R(u_1,\ldots, u_k)$
of $\phi$
such that $(f(u_1),\ldots,f(u_k)) = (b_1, \ldots, b_k)$.
Clearly, each \emph{yes} instance of $\cond(\relb)$
is also a \emph{yes} instance of $\scsp(\relb)$.
In graph-theoretic settings, a related notion has been studied:
a \emph{compaction} is defined similarly, but typically
is not required to map onto self-loops of the target graph $\relb$.

There exists an elementary argument that,
for each structure $\relb$, the problem
$\scsp(\relb)$ polynomial-time Turing reduces to the problem
$\cond(\relb)$; see~\cite[Proof of Proposition 1]{BodirskyKaraMartin12-surjectivesurvey}.\footnote{
The proof there concerns the compaction problem,
but is immediately adapted to the present setting.}
We observe here a condition that allows one to show 
NP-hardness 
for the problem $\cond(\relb)$,
under polynomial-time many-one reduction,
 using the developed framework.

\begin{theorem}
\label{thm:cond-hardness}
Let $F$ be an encoding that contains
all constants in the sense that, for each $b \in B$,
there exists a map $f_b \in F$ such that
$f_b(e) = b$ for each $e$ in the domain of $f_b$.
Suppose that $\relb$ is $F$-stable, and that
$\reld$ is an $F$-induced template of $\relb$.
Then, the problem $\csp(\reld)$ polynomial-time
many-one reduces to $\cond(\relb)$.
\end{theorem}

\begin{proof}
The reduction is that of Theorem~\ref{thm:stable-gives-reduction}.
We use the notation from that proof.
It suffices to argue that when $(A_{V,F},\psi)$
is a \emph{yes} instance of $\scsp(\relb)$,
it is also a \emph{yes} instance of $\cond(\relb)$.
From that proof,
when there exists a surjective satisfying assignment of $\psi$,
there exists such a surjective satisfying assignment 
of the form $t[g]$.  
Since $t[g]$ maps each application of the form
$(\cdot,f_b)$ to $b$, we obtain that $(A_{V,F},\psi)$
is a \emph{yes} instance of $\cond(\relb)$.
\end{proof}

Each encoding presented in 
Section~\ref{sect:hardness-results} 
contains all constants; thus, via
Theorem~\ref{thm:cond-hardness}, we obtain that,
for each structure $\relb$ treated in that section,
the problem $\cond(\relb)$ is NP-complete.

\section{Sparsifiability}
\label{sect:sparsifiability}

The framework presented here also has implications
for the sparsifiability of the problems $\scsp(\relb)$.
We explain the main ideas here;
our presentation is based on the 
framework and terminology of~\cite{ChenJP20-sparsifiability-journal}.

For a problem $Q$ of the form $\scsp(\relb)$
or of the form $\csp(\relb)$, 
we define a 
\emph{generalized kernel} as a polynomial-time many-one reduction $f$ from $Q$ to any other decision problem $Q'$.  
We say that a generalized kernel $f$ has 
\emph{bitsize $h: \N \to \N$} 
when, for any instance $x$ of $Q$
having $n$ variables, it holds that
$|f(x)| \leq h(n)$.
That is, we grade each generalized kernel
according to the image size of each
instance of $Q$, measured as a function
of the number of variables of $Q$.

When one has an encoding $F$ of maps each having arity $\leq 1$, and the hypotheses of Theorem~\ref{thm:stable-gives-reduction} hold, this theorem's reduction from
$\csp(\reld)$ to $\scsp(\relb)$
increases the number of variables at most linearly:
for each instance of the first problem having
$n \geq 1$ variables, the instance's image
under the reduction has at most $|F|n$ variables.
This reduction thus constitutes a 
\emph{linear-parameter transformation},
in the parlance of
\cite[Definition 2.2]{ChenJP20-sparsifiability-journal}.  Such transformations
allow one to transfer lower bounds from
the first problem (in our case, $\csp(\reld)$)
to the second problem (in our case, $\scsp(\relb)$).

We establish the following new result.

\begin{theorem}
	The problem $\scsp(\reln)$ (studied in Section~\ref{subsect:no-rainbow})
has a generalized kernel of bitsize $O(n^2 \log n)$, but no 
generalized kernel of bitsize $O(n^{2-\epsilon})$, for any $\epsilon > 0$, under the assumption that NP is not in coNP/poly.
\end{theorem}

\newcommand{\nae}{\mathbf{B}_{\mathsf{NAE}}}

\begin{proof}
Let $\nae$ be a structure 
having as its sole relation
the not-all-equal relation of Theorem~\ref{thm:no-rainbow-np-complete}.
The problem $\csp(\nae)$ does not have a generalized kernel of bitsize $O(n^{2-\epsilon})$ (for any $\epsilon > 0$), under the assumption that NP is not in coNP/poly;
this follows from Theorem 3.10 of~\cite{ChenJP20-sparsifiability-journal}. 
From the linear-parameter transformation
from $\csp(\nae)$ to $\scsp(\reln)$,
one obtains the lower bound of the theorem statement.

It remains to argue that the problem $\scsp(\reln)$ has a generalized kernel of bitsize $O(n^2 \log n)$.
 To show this, we use an adaptation of the
polynomial method pioneered by Jansen and Pieterse~\cite{JansenPieterse19-optimal-sparsification-polynomials}; see also~\cite{ChenJP20-sparsifiability-journal}.
Consider an instance $(V,\phi)$ of the problem $\scsp(\reln)$.
Each constraint has the form $R(v_1,v_2,v_3)$,
where $v_1,v_2,v_3 \in V$; for each constraint,
we create an equation 
$e_{\{v_1,v_2\}} + e_{\{v_1,v_3\}} + e_{\{v_2,v_3\}}
- 1 = 0$, to be interpreted over $F_2$,
the field with $2$ elements.
Here, the variables used 
in equations are those in the set
$\{ e_{\{v,v'\}} ~|~ v, v' \in V, v \neq v' \}$,
which is 
a set of ${|V| \choose 2}$
variables.  For any assignment 
$f: V \to \{ 0, 1, 2 \}$, set the variable
$e_{\{v,v'\}}$ to be $1$ if $f(v) = f(v')$,
and to be $0$ otherwise; then, the assignment
satisfies $R(v_1,v_2,v_3)$ if and only if
its corresponding equation is true over $F_2$.
The solution space of the
created equations thus capture the set of
satisfying assignments.  However, since there are
$O(n^2)$ variables used in these equations, 
by linear algebra, 
a polynomial-time algorithm may find a subset
of these equations having size $O(n^2)$ that has
the same solution space as the original set.
The generalized kernel outputs the constraints
corresponding to these equations; each
constraint can be written in $O(\log n)$
space, and so the claimed upper bound follows.
\end{proof}

\paragraph{Acknowledgements.}  
The author thanks Barny Martin, Dmitriy Zhuk,
and Benoit Larose for
useful discussions; in particular,
the former two are thanked for explanations
concerning their results~\cite{MartinPaulusma15-discut,Zhuk20-norainbow}, respectively.  The author also thanks Victor Lagerkvist and Moritz M\"uller 
for helpful comments.

\bibliographystyle{plain}

\bibliography{/Users/hubie/Dropbox/active/writing/hubiebib.bib}

\end{document}